\documentclass[11pt]{article}

\def\shownotes{1} 

\usepackage{xspace}
\usepackage{graphicx}
\usepackage{cite,url}
\usepackage{amsthm}
\usepackage{amsfonts}
\usepackage{fullpage}

\ifnum\shownotes=1
\newcommand{\authnote}[2]{{ $\ll$\textsf{\footnotesize #1 notes: #2}$\gg$}}
\else
\newcommand{\authnote}[2]{}
\fi

\providecommand{\ie}{\emph{i.e.,} }
\providecommand{\eg}{\emph{e.g.,} }

\providecommand{\etc}{\emph{etc.}}      

\newtheorem{thm}{Theorem}[section]

\newtheorem{lem}[thm]{Lemma}

\newtheorem{lemma}{Lemma}
\newtheorem{definition}{Definition}

\begin{document}

\title{Network-Destabilizing Attacks\thanks{This is the revised and expanded version of a paper that appeared as a brief announcement in \emph{PODC} 2012 \cite{LGSPODC12}.  This project was supported by NSF Grant S-1017907.}}
\date{}

\author{Robert Lychev \and Sharon Goldberg \and Michael Schapira}
\maketitle

\begin{abstract}
The Border Gateway Protocol (BGP) sets up routes between the smaller networks that make up the Internet. Despite its crucial role, BGP is notoriously vulnerable to serious problems, including (1) propagation of bogus routing information due to attacks or misconfigurations, and (2) network instabilities in the form of persistent routing oscillations.  The conditions required to avoid BGP instabilities are quite delicate.  How, then, can we explain the observed stability of today's Internet in the face of common configuration errors and attacks?

This work explains this phenomenon by first noticing that almost every observed attack and misconfiguration to date shares a common characteristic: even when a router announces egregiously bogus information, it will continue to announce the same bogus information for the duration of its attack/misconfiguration. We call these the ``fixed-route attacks'', and show that, while even simple fixed-route attacks can destabilize a network, the commercial routing policies used in today's Internet prevent such attacks from creating instabilities.

\end{abstract}

\section{Introduction}

The Internet is composed of smaller networks, called Autonomous Systems (ASes) (\eg AT\&T, Bank of America, Google, \etc). The Border Gateway Protocol (BGP) is a distributed protocol that allows ASes to learn to reach distant ASes via announcements from their neighboring ASes. Each BGP announcement contains a list of every AS en route to a destination; each AS repeatedly applies its local routing policy to select a single available route to each destination, and announces that route to its neighbors. Despite it's crucial role, BGP routing is notoriously vulnerable to a number of serious problems:

\smallskip\noindent{\bf Bogus routing information.} Because the Internet currently lacks infrastructure to validate the correctness of information in routing messages (\eg does the route actually exist? is one AS impersonating another), an AS can propagate false routing information through the Internet and thus influence the routes selected by other ASes.
We see this quite frequently in practice; a typical cause is a configuration error, where a router is mistakenly programmed to announce a bogus route\cite{AS7007,con-ed}, but we also worry about attacks where a router deliberately manipulates routing information to draw traffic to its network (that it can tamper with, drop, or eavesdrop on)~\cite{DEFCON,pakistan,china}.

\smallskip\noindent{\bf Instability.} BGP allows ASes great expressiveness in configuring local routing policies. Unfortunately, these routing policies can interact in ways that lead to persistent routing oscillations, \ie situations where some ASes endlessly change the route they select, even when network structure is static (in terms of network topology, ASes' routing policies, \etc). BGP oscillations render the network unpredictable and can significantly harm network performance, as every time a router switches routes, it is bound to delay, misorder, or even drop, some fraction of the traffic it is carrying) ~\cite{katabiVoIP,labovitzStability}.

\smallskip\noindent
On the bright side, we have never seen events in which bogus routing information has inadvertently lead to a BGP instability. One might claim the attacks/misconfigurations we have seen in the wild were never intended to create BGP instabilities. However, given the delicate conditions required to avoid BGP instabilities \cite{GSW,GR}, the fact that a misbehaving AS has never caused the system to tip into an unstable state is quite surprising.  How, then, can we explain the observed stability of today's Internet in the face of common errors and attacks?

This work explains this phenomenon by first noticing that almost every observed attack and misconfiguration to date~\cite{AS7007,con-ed,pakistan,china,DEFCON} shares a common characteristic: even when a router announces egregiously bogus information, it will continue to announce the same bogus information for the duration of its attack/misconfiguration.  We call this class of attacks the ``fixed-route attacks''; one famous and common example of this attack is the ``prefix hijack'', where an AS announces an IP prefix belonging to another AS (\eg in 2008 Pakistan Telecom claimed to be the legitimate destination for Internet addresses belonging to YouTube, resulting in YouTube-bound traffic reaching Pakistan Telecom instead\cite{pakistan}).

While it is quite easy to come up with examples where a single fixed-route attack destabilizes BGP (see Figure~\ref{fig:gadget}),  our main result is to show that the \emph{routing policies} used in today's Internet prevent such attacks from triggering instabilities.

\subsection{Our Model}


We now present a brief and intuitive exposition of our model, which is based on the seminal work of Griffin \emph{et al.}~\cite{GSW} on BGP stability. See Section~\ref{sec:model} for a more thorough explanation. The network is modeled as an undirected graph $G=(V,E)$, where the set of nodes (vertices) $V$ represents the ASes, and the set of edges $E$ represents BGP communication links between ASes. The vertex set $V$ contains a unique destination node $d$ to which all other nodes aim to establish routes.\footnote{This is the standard formulation~\cite{GSW}.  We use this because BGP establishes routes to every destination IP prefix independently.} The routing system evolves over an infinite sequence of discrete time steps $t=1,\ldots$ At each time step $t$ a subset of the nodes is ``activated''. Whenever a non-attacker node is activated it executes the following actions:

\begin{enumerate}

\item Process the most recent update messages received from neighboring nodes, where each message contains the explicit list of all nodes on the neighbor's route to the destination.

\item Select the single ``best'' available route according to a local \emph{ranking} of all simple (loop-free) routes to the destination.

\item Announce this route to a subset of the neighboring nodes via update messages according to a local ``\emph{export policy}'', which determines which routes the node is willing to make available to each of its neighbors.

\end{enumerate}

Whenever a (fixed-route) attacker node is activated, it announces a fixed route (list of nodes ending in $d$) to each of its neighbors. We stress that, apart from the requirement that the attacker repeatedly send the same BGP route announcement to each neighbor, no restrictions whatsoever are imposed on the attacker. Thus, the attacker can pretend to be the destination (simply announce ``d''), announce different paths of nodes to different neighbors, not announce any route to some neighbors, \etc

Our aim is to identify conditions which imply network stability, where by network stability we mean the guarantee that from some moment forth, every non-attacker node's chosen route remain fixed, for \emph{every} choice of initial state of the system and schedule of node activation and update message arrivals.\footnote{We consider ``fair'' schedules. Intuitively, a schedule is ``fair'' if no node is indefinitely starved from acting, or from receiving update messages from a neighboring node. We stress that update messages in our model can be arbitrarily delayed and even dropped. Our positive results do not rely on any assumptions on the order of update message arrivals (\eg FIFO queueing).}

\subsection{Our Results}

\subsubsection{Even a Single Fixed-Route Attacker Can Destabilize a Network!}

Simple examples show that a stable network can easily be rendered unstable even by a single fixed-route attacker. Consider, for instance, the network described in Figure~\ref{fig:gadget}, and assume that each node's ranking of routes is as depicted beside it, and that each node has the same route-export policy---export all routes to all neighbors.

\begin{figure}
 \centering
 \includegraphics[width=1.7in]{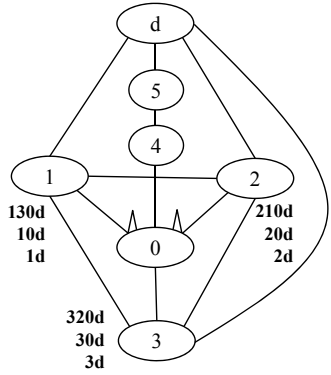}\caption{Node 0 can destabilize this network.} \label{fig:gadget}
\end{figure}

Now, consider this network \emph{before} node $0$ launches an attack. Even though each of nodes $1,2,$ and $3$ prefers the longer routes to $d$ via node $0$, these routes will not become available as the link $(0,d)$ does not exist. Thus, each of these nodes will choose the direct route to $d$, and the network is stable. Now, suppose node $0$ launches a fixed-route attack by announcing the bogus route ``$0,d$'' to all of its neighbors. This network is now an instance of the classic \textsc{Bad Gadget} network~\cite{GSW}, which is notoriously unstable! To understand why, suppose that nodes $1$ and $2$ think they are routing along $2,1,0,d$, while node $3$ thinks it uses the route $3,0,d$. This is clearly unstable, since node $1$ would rather be using the route $1,3,0,d$, and so it will change its route selection. By symmetry, this situation will repeat endlessly.  (Of course, each of nodes $1,2,3$ are actually using a much longer route through nodes $4$ and $5$, but because the Internet lacks the infrastructure to validate the correctness of routes, they have no idea that this is the case.)

\smallskip We identify interesting environments where stability is maintained in the presence of fixed-route attackers. We also quantify convergence rate in terms of asynchronous rounds~\cite{DolevT06,sss}, \ie periods of time in which each node gets at least one update message from each neighbor, and is activated at least once after receiving these updates.

\subsubsection{Shortest-Path Routing is Stable in the Presence of Fixed-Route Attacks}

We first consider the simple scenario that all non-attackers have shortest-path rankings of routes, that is, always prefer shorter routes to longer ones. We present the following, relatively easy to prove, result, which holds for all choices of export policies:

\vspace{0.1in}\noindent {\bf Theorem 1:} {\em When all nodes have shortest-path rankings, convergence to a stable routing state is guaranteed within $|V|$ asynchronous rounds even in the presence of fixed-route attacks.}\vspace{0.1in}

We point out that this positive result holds for every network, regardless of the number and locations of the fixed-route attackers, and of the specific fixed-route attacks launched.

\subsubsection{Commercial Routing is Stable in the Presence of Fixed-Route Attacks}

We now turn our attention to the commercial routing framework presented by Gao and Rexford~\cite{GR}, which is believed to capture ASes' routing policies in practice. ASes sign bilateral long-term business contracts which determine who provides connectivity to whom. Typically, neighboring ASes have one of two business relationships: \emph{customer-provider}, in which one AS (the customer) purchases connectivity from another AS (the provider), and \emph{peering}, in which the two ASes agree to carry transit traffic between their customers for free. These business relationships naturally induce the following restrictions on ASes' routing policies, formalized by Gao and Rexford in \cite{GR}: These business relationships naturally induce restrictions on ASes' routing policies: (1) an AS prefers revenue-generating routes through customers over routes through its peers and providers; and (2) an AS only carries traffic from one neighbor to another neighbor if at least one of them pays it, \ie is its customer.\footnote{We also assume that there can be no cycle of customer-provider edges in the AS-level digraph, as an AS cannot be an indirect customer of itself\cite{GR}.}

We explore network stability in the context of commercial routing. Our main, surprisingly strong, result is the following:

\vspace{0.1in}\noindent {\bf Theorem 2:} {\em When all nodes have commercial routing policies, convergence to a stable routing state is guaranteed within $2X+1$ asynchronous rounds even in the presence of fixed-route attacks, where $X$ is the depth of the customer-provider hierarchy.}\vspace{0.1in}

Again, this result holds regardless of the number and locations of the fixed-route attackers, and of the specific fixed-route attacks launched. In today's Internet, the depth of the customer-provider hierarchy is very shallow (roughly five levels on average~\cite{cyclops}). Hence, commercial routing guarantees not only network stability, but also fast convergence, even in the presence of fixed route attacks. In our view, these results explain the observed stability of today's Internet in the face of common configuration errors and attacks.

\section{Model}\label{sec:model}

\subsection{BGP Dynamics in the Presence of Fixed-Route Attacks}

\noindent{\bf Network and routing policies.} The network is defined by an \emph{AS graph} $G=(V,E)$, where $V$ represents the set of ASes, and $E$ represents BGP communication links between ASes. $V$ consists of $n$ \emph{source-nodes} $\{1, \ldots, n\}$, and a unique \emph{destination node} $d\notin [n]$. (We follow the standard model~\cite{GSW} where there is a single destination node, because in the Internet routes to every destination IP prefix are computed independently.) Let $S_A\subseteq V$ denote the set of \emph{fixed-route attackers}.

We think of a route in $G$ from a node $i$ to the destination node $d$ as a sequence of nodes starting in node $i$ and ending in node $d$. For every non-attacker source node $i\in [n]\setminus S_A$, let $X^i$ denote the set containing all possible sequences of nodes in starting in node $i$ and ending in the destination $d$, and also the ``empty route'' $\emptyset$.  Each non-attacker source node $i$ has a \emph{ranking function} $\leq_i$, which defines an order over all sequences of nodes in $X^i$. Note that $i$'s preferences over routes in our model are thus not restricted only to routes which actually exist in the network, but also to non-existent routes. This is needed to model $i$'s reaction to bogus routing information propagated by an attacker. The empty route captures the possibility that $i$ has no route to the destination. We allow ties in $\leq_i$ between two sequences in $X^i$ only if they share the same ``next-hop'' node, \ie the node that comes after $i$ in both sequences is the same.

%

In addition, each non-attacker source node $i$ has an {\em export policy}, which captures the routes that $i$ is willing to announce to neighboring nodes. $i$'s export policy specifies, for every neighboring node $j$, a set of routes $E^{ij}\subseteq
X^i$, such that $\emptyset\in E^{ij}$, that $i$ is willing to announce to neighbor $j$.

\vspace{0.1in}\noindent{\bf System dynamics.} BGP belongs to an abstract family of routing protocols named \emph{path-vector protocols}~\cite{GSW}. Basically, the routing tree to a given destination is built, hop-by-hop, as knowledge about how to reach that destination propagates through the network. The process is initialized when $d$ sends an ``update message'' to announces itself. From this moment forth, each (``well behaved'') source node establishes a route to $d$ by repeatedly choosing the ``best'' route announced to it by its neighbors and announcing this route to its neighbors.

This modeled as follows. The routing system evolves over an infinite sequence of discrete time steps  $t=1\ldots$. At each time step $t$, a subset of the nodes $S^t\subseteq V$ is ``activated''. Whenever a source node $i\notin S_A$ is activated, it executes the following actions:

\begin{enumerate}

\item Process the most recent update messages received from neighboring nodes, where each message contains the explicit list of all nodes on the neighbor's route to the destination.

\item Select the single ``best'' available route according to the local ranking routes $\leq i$.

\item Announce this route to a subset of the neighboring nodes via update messages according to node $i$'s export policy.

\end{enumerate}

Whenever an attacker node $A\in S_A$ is activated, it announces a fixed list of nodes ending in the destination $R_{Ai}$ to each neighboring nodes $i$ via BGP update messages.

\noindent{\bf Convergence in the presence of fixed-route attackers.} We say that the routing system is \emph{convergent} if for every initial routing configuration, and for every possible order of node activations and update message arrivals, then there is a point in time after which the selected route of each non-attacker source node remains constant. We use the notion of \emph{asynchronous rounds} to measure the convergence rate of the system (worst-case across all initial routing configurations and all fair schedules). An asynchronous round is a period of time in which each node gets at least one update message from each neighbor (if there are update messages in transit), and is activated at least once after receiving all these updates. The first asynchronous round is the shortest period of time in which this is satisfied. The second asynchronous round is the shortest period of time in which this is satisfied immediately following the first asynchronous round,\emph{ etc.}

\subsection{Routing Systems}

\subsubsection{Shortest-Path Routing Systems}

Source node $i$ has a \emph{shortest-path} ranking if $i$ always prioritizes shorter routes over longer routes. We measure route length by the number of nodes on that route. Thus, for every two non-empty routes $Q,R\in X^i$, if $R$ is shorter than $Q$, then $Q <_i R$. We make no restrictions on how routes of equal length are ranked. A shortest-path routing system is a routing system where all source nodes have shortest-path rankings. No restrictions are imposed on nodes' export policies in shortest-path routing systems.

\subsubsection{Commercial-Routing Systems}

In today's commercial Internet neighboring ASes typically have one of two common business relationships: \emph{customer-provider} and \emph{peer-to-peer}. These relationships induce a hierarchy, in which no AS is its own indirect customer (that is, there is no sequence of ASes $i_1,\ldots,i_k$ such that $i_1=i_k$ and every AS is a customer of the AS that comes before it). We now present ranking functions and export policies that are consistent with the economics of these relationships.

\vspace{0.1in}\noindent{\bf Commercial rankings.} Consider a specific source node $i$. We call a route in $X^i$ a ``\emph{customer route}'' if the first edge on the route is from $i$ to a customer of $i$. Similarly, we call a route in $P^i$ a ``\emph{peer route}'', or a ``\emph{provider route}'', if the first edge on the route is from $i$ to its peer, or provider, respectively.  Source node $i$ has a commercial ranking~\cite{GR} if it always prefers customer routes over peer and provider routes, that is, for every customer route $Q$ and every peer or provider route $R$ it holds that $R<_i Q$.

\vspace{0.1in}\noindent{\bf Commercial export policies.} Source node $i$ has a commercial export policy if it exports peer or provider routes only to its customers.  That is, for every neighboring node $j$ that is a peer or provider of $i$, it follows that $E^{ij}$ consists of customer routes only.

\section{Shortest-Path Routing Systems}

We now present the notion of perceivable routes, which plays a major role in our proofs. We first introduce the following notation.
Consider a route $R =(i,\ldots,j,\ldots,d)$. We denote the prefix of route $R$ ending in node $j$, and the suffix of route $R$ starting in node $j$, by $R^{|j}$ and $R_{|j}$, respectively. We denote the predecessor node and successor node of node $j$ on $R$ by $pred(j,R)$ and $succ(j,R)$, respectively.

Intuitively a route $R\in P^i$ is perceivable at node $i$ if there is some imaginable scenario in which this route is propagated, hop by hop, towards $i$ from either the destination node $d$ or an attacker node $A\in S_A$.

\begin{definition} [perceivable routes]
A simple (loop-free) route $R = \{i, \ldots, d\}\in P^i$ is \emph{perceivable} at node $i\in [n]$ if one of the two following conditions holds:

\begin{enumerate}

\item Route $R$ contains no fixed-route attackers in $S_A$, and $R_{|k}\in E^{k, j}$ for every edge $(j, k)\in R$.

\item Route $R$ contains a single fixed-route attacker $A\in S_A$ such that (i) $E^{A, pred(A,R)} = R_{|A}$; and (ii) $R_{|k}\in E^{k, j}$ for every edge $(j, k)\in R^{|pred(A,R)}$.

\end{enumerate}
\end{definition}

\noindent Let $PR^i$ to be the set of all perceivable routes at node $i$.  We introduce another concept before proving the main theorem in this section.

\begin{definition}[best perceivable routes]
The set of \emph{best perceivable routes} of a non-adversarial node $i$, $BPR^i \subseteq PR^i$, is the set of $i$'s perceivable routes that have the highest rank in $\leq_i$.
\end{definition}

Recall that ties in $\leq_i$ are only allowed between routes that share the same next-hop node. Thus, the set of best perceivable routes for a node $i$ must all share the same next hop. We call this node the ``\emph{best next hop}'' for $i$.

\begin{thm} \label{thm_SP_convergence_time}
Every shortest-path routing system is convergent for every choice of fixed-route attackers in $N$.  Moreover, convergence to a stable state is guaranteed within $n$ asynchronous rounds.
\end{thm}

\begin{proof} Before diving into the proof, we provide the reader with intuition.  Our strategy is to design an iterative algorithm called \textbf{Fix Shortest Routes} (FSR) that fixes nodes to the shortest routes that are available to those nodes abiding by all aspects of export policies and consistency with respect to non-adversarial nodes.  With every iteration of this algorithm, we fix a single node to a route as explained above and add that node to a set of $\mathcal{I} \subseteq V$ which consists only of nodes who would immutably announce the same route once they have been added to that set.  Note that at the very beginning of this algorithm, $\mathcal{I}$ contains only fixed-route attackers and $d$.

We first describe FSR.  It starts with $\mathcal{I}$ containing only $S_A$ and $d$.  In FSR, while there is at least one node $i \notin \mathcal{I}$ such that $BPR^i$ is not empty, the following steps are repeated:
\begin{enumerate}
	\item{over all $i \notin \mathcal{I}$ for which $PR^i \neq \emptyset$, select an arbitrary node $i \notin \mathcal{I}$ with the shortest route in $PR^i$ such that $Nxt^i \in \mathcal{I}$;}
 	\item{add $i$ to $\mathcal{I}$;}
 	\item{from all nodes' $PR$ sets, remove all routes  that contain $i$ but that violate $i$'s export policies with respect to or whose suffix at $i$ is not the route $i$ is fixed to;}
 	\item{all nodes with empty $PR$ sets get added to $\mathcal{I}$.}
\end{enumerate}
FSR outputs all nodes in $\mathcal{I}$ and their stable routes, excluding fixed-route attackers and $d$.

First notice that FSR fixes every node to a route (possibly the empty route $\emptyset$). We prove the theorem by proving the following lemma:

\begin{lemma}
For any arbitrary non-adversarial node $i$ who gets fixed to a non-empty route of length $k$ with FSR, $i$ will stabilize to the same route within $k$ asynchronous rounds from any initial routing configuration.
\end{lemma}
\begin{proof}
We prove this lemma by induction. 

\vspace{0.05in}\noindent{\bf Induction hypothesis.} For any initial routing configuration, for any arbitrary non-adversarial node, if that node does not stabilize to a route after $k$ asynchronous rounds, then after $k$ asynchronous rounds it will not announce a route of length at most $k$ to any of its neighbors, but if that node gets fixed to a route of length at most $k$ with FSR, then that node will stabilize to the same route within $k$ asynchronous rounds.

\vspace{0.05in}\noindent{\bf Base of the induction.} Consider an arbitrary routing configuration and suppose that all fixed-route attackers in $S_A$ announce their fixed routes from the start.  Suppose that, W.L.O.G., there is at least one non-attacker node in the routing system whose $PR$ set consists of a route of length $1$.  For the base case, consider an arbitrary, node $i\notin S_A$,  whose routes in $BPR^i$ are of length $1$, so this node gets fixed to a one-hop route with FSR by construction.  During BGP's execution, regardless of whether $i$'s shortest routes contain a fixed-route attacker or not, $i$ must learn of all such routes within a single asynchronous round (since it is only one-hop away from $d$ or a fixed-route attacker (running a prefix-hijacking attack), or both).  Notice that this is true for any initial routing configuration because after one asynchronous round all nodes who thought they were one hop away from $d$ mistakenly, will have realized their misconception and have stopped announcing one-hop routes to their neighbors, while nodes who mistakenly thought they were farther away from $d$ than one hop, will have realized that they have direct perceivable routes.  This holds for any activation schedule since all perceivable, one-hop routes will have been revealed to all such nodes within a single asynchronous round.  Thus, after a single asynchronous round, $i$ is guaranteed to stabilize to its best, shortest route in its $PR^i$, which is exactly what it is fixed to with FSR.  

\vspace{0.05in}\noindent{\bf Induction step.}  Consider an arbitrary node $j$ who gets fixed to a route of length $k+1$ with FSR.  $j$'s next hop on this route must have been fixed to a $k$-hop route by FSR which, by induction hypothesis, must be the same route that node converges to within $k$ asynchronous rounds.

Within the $(k+1)^{\textit{th}}$ asynchronous round, $j$ must receive an announcement from all its neighbors who have stabilized to their best, shortest routes of length $k$ and who are willing to export their routes to $j$.  $j$ could also learn of routes from neighbors who do not stabilize to routes within $k$ asynchronous rounds.  However, by induction hypothesis, such neighbors would have to be announcing routes of length greater than $k$ and be fixed to routes longer than $k$ in FSR, so $j$ can ignore these routes since they are not in $BPR^j$.  Note that since step 4 of FSR ensures that nodes get fixed to routes that are exportable and consistent with respect to the routes that other nodes are being fixed to and in non-decreasing order of route length, during BGP's execution $j$ will learn of exactly the same routes from these neighbors during $(k+1)^{\textit{th}}$ asynchronous round as the shortest routes in $PR^j$ when $j$ is being fixed to its route during FSR.  Hence, within $k+1$ asynchronous rounds, $j$ will have learned of all shortest routes that could be available to $j$ considering export rules and route consistency.  This must hold for any initial routing configuration because after the $(k+1)^{\textit{th}}$ asynchronous round, all nodes who thought they were $k+1$ hops away from $d$ mistakenly, will have realized their misconception and have stopped announcing routes of length at most $k+1$ to their neighbors.  This is true because all nodes who have stabilized to $k$-hop routes will have announced their routes to all appropriate neighbors by the end of this round while those who have not stabilized to $k$-hop routes will stop announcing routes of length at most $k$ by induction hypothesis.  For the same reason, nodes who mistakenly thought they were farther away from $d$ than $k+1$  hops, will have learned of all their $(k+1)$-hop perceivable routes .  Note that this must hold for any activation schedule, since all perceivable $(k+1)$-hop routes will have been revealed to all such nodes within the $(k+1)^{\textit{th}}$ asynchronous round, at which point  $j$ is guaranteed to stabilize to its most preferred route in $BPR^j$ via a tie-breaking rule, which is exactly the route it is fixed to with FSR.

\end{proof}

The Theorem then follows because no node can have a $PR$ set with shortest routes of length greater than $n$.

\end{proof}

\section{Commercial Internet Routing Systems}
Recall that perceivable routes of any node $i \notin S_A$ are the routes that $i$ could ever perceive to exist if they are propagated, hop by hop, towards $i$ from either $d$ or an attacker $A \in S_A$.  Also, recall that
ties between perceivable routes of any node $i \notin S_A$, occur only when the next hop of those routes is the same, in which case $i$ is allowed to select its most preferred route using an arbitrary but consistent tie-breaking rule.

\begin{thm} \label{thm_GR_convergence_time}
Every commercial-routing system is convergent for every choice of fixed-route attackers in $S_A \subseteq V$.  Moreover, convergence to a stable state is guaranteed within (2x+1) asynchronous rounds, where x is the height of the customer-provider hierarchy.
\end{thm}

\begin{proof}
The proof follows from Lemmas \ref{lem_get_stable_provider}-\ref{lem_stabilize_customers} below.
\end{proof}

Before delving into the proof, we provide the reader with an outline as follows.  Our strategy is to design an iterative algorithm called \textbf{Fix Routes} (FR) that fixes nodes to the best routes that are available to those nodes abiding by all aspects of the commercial routing and consistency with respect to non-adversarial nodes.  With every iteration of this algorithm, we fix a single node to a route as explained above and add that node to a set of $\mathcal{I} \subseteq V$ which consists only of nodes who would immutably announce the same route once they have been added to that set.  Note that at the very beginning of this algorithm, $\mathcal{I}$ contains only fixed-route attackers and $d$.

FR consists of three subroutines: \textbf{Fix Customer Routes} (FCR), \textbf{Fix Peer Routes} (FPeeR), and \textbf{Fix Provider Routes} (FPrvR), that FR executes in that order.  By adding them one-by-one to $\mathcal{I}$, FR fixes all nodes with customer, peer and provider routes with FCR, FPeeR and FPrvR respectively.  In Lemma \ref{lem_convergence}, we then show that every node in $\mathcal{I}$ thus-constructed, in fact stabilizes to the same route in BGP as the route that node has been fixed to by FR.  We conclude with Lemmas \ref{lem_stabilize_providers}-\ref{lem_stabilize_customers}, which collectively show that every node converges to that route within $2x+1$ asynchronous rounds.

We now describe FR and its subroutines.  FR starts with FCR; at this point $\mathcal{I}$ contains $S_A$ and $d$.  In FCR, while there is at least one node $i \notin \mathcal{I}$ such that $BPR^i$ contains at least one customer route, the following steps are repeated:
\begin{enumerate}
	\item{select an arbitrary node $i \notin \mathcal{I}$ such that $BPR^i$ contains a customer route;}
 	\item{find a node $j \notin \mathcal{I}$ whose most preferred route in $BPR^j$ is a customer route and $Nxt^j \in \mathcal{I}$;}
 	\item{add $j$ to $\mathcal{I}$;}
 	\item{from all nodes' $PR$ sets, remove all routes  that contain $j$ but whose suffix at $j$ is not what $j$ is fixed to above;}
 	\item{all nodes with empty $PR$ sets get added to $\mathcal{I}$.}
\end{enumerate}
FR outputs all nodes in $\mathcal{I}$ and their stable routes, excluding fixed-route attackers and $d$.

Before proceeding any further, in Lemma \ref{lem_get_stable_provider} we show that in step 2 of FCR we are guaranteed to find a node whose most preferred route in its $BPR$ set contains a customer route only one hop away from a node in $\mathcal{I}$ as long as there is at least one node $i \notin \mathcal{I}$ such that $BPR^i$ contains a customer route.

\begin{lem} \label{lem_get_stable_provider} Consider an arbitrary node $i \notin S_A$ such that $BPR^i$ contains a customer route $R$.  There exists at least one node $j \notin \mathcal{I}$ whose most preferred route in $BPR^j$ is a customer route and $Nxt^j \in \mathcal{I}$. \end{lem}

\begin{proof}
We consider an arbitrary routing configuration and suppose that all fixed-route attackers announce their fixed routes from the start.
If $i$ is fixed to $R$, then we are done.  Otherwise, $R$ must not always be available to $i$, so there must be some node $j$ in $R$ who prefers some other route $R'$ to its suffix of $R$.  Consider the closest such node $j$ to $i$, suppose that $R'$ is $j$'s most preferred route, and note that $R'$ must be a customer route.  If $j$ is fixed to $R'$, then we are done.  If not, then $R'$ must not be always available to $j$, so there must be at least one node $k$ on $R'$ that prefers some route $R''$ to its suffix of $R'$.  Consider the closest such node $k$ to $j$, suppose that $R''$ is $k$'s most preferred route, and note that $R''$ must be a customer route.  If $k$ is fixed to $R''$, then we are done.  Otherwise, $R''$ must not be always available to $k$, so we can find at least one node on that route whose suffix is not that node's most preferred route.  Since $n$ is finite, we can continue this argument until we either form a customer-provider cycle or we reach a node whose $BPR$ set contains a direct customer route to a node in $S$ that it is fixed to.  Note that the customer-provider cycle that we obtain in the former case cannot contain a fixed-route attacker since this cycle was obtained by repeatedly considering nodes whose most preferred routes were not always available to them (the routes that fixed-route attackers announce are always available to the neighbors they are being announced to).  Therefore, only the latter case is possible since our routing systems do not allow for a node to be its own indirect customer.
\end{proof}

We now describe FR's operations during execution of FPeeR.  This subroutine starts with $\mathcal{I}$ and the configuration of the routing system the way it is after execution of FCR, i.e. $\mathcal{I}$ contains only fixed-route attackers, $d$, and nodes with empty and customer routes.   In FPeeR, while there is at least one node $i \notin \mathcal{I}$ such that $BPR^i$ contains at least one peer route with the next hop being in $\mathcal{I}$, the following steps are repeated:
\begin{enumerate}
	\item{select an arbitrary node $i \notin S$ such that $Nxt^i$ is in $\mathcal{I}$, and either has a customer route or is $d$ or a fixed-route attacker;}
 	\item{find a node $i \notin \mathcal{I}$ whose most preferred route in $BPR^i$ is a peer route and $Nxt^i \in \mathcal{I}$;}
 	\item{from all nodes' $PR$ sets, remove all peer and customer routes  that contain $i$, and remove all provider routes whose suffix at $i$ is not what $i$ is fixed to above;}
 	\item{all nodes with empty $PR$ sets get added to $\mathcal{I}$.}
\end{enumerate}
\noindent FPeeR outputs all nodes in $\mathcal{I}$ and their stable routes, excluding fixed-route attackers, $d$ and all nodes output by FCR.

We now describe FR's operations during execution of FPrvR.  This subroutine starts with $\mathcal{I}$ and the configuration of the routing system the way it is after the consecutive execution of FCR and FPeeR in that order, i.e. $S$ contains fixed-route attackers, $d$, nodes with empty, customer and peer routes.   In FPrvR, while there is at least one node that is not in $\mathcal{I}$ and with a provider route in its $BPR$, the following steps are repeated:
\begin{enumerate}
	\item{select an arbitrary node $i \notin S$ such that $BPR^i$ contains a provider route;}
	 \item{find a node $j \notin \mathcal{I}$ whose most preferred route in $BPR^i$ is a provider route and $Nxt^j \in \mathcal{I}$;}
 	\item{add $j$ to $\mathcal{I}$;}
 	\item{from all nodes' $PR$ sets, remove all peer and customer routes  that contain $j$, and remove all provider routes whose suffix at $j$ is not what $j$ is fixed to above;}
 	\item{all nodes with empty $PR$ sets get added to $\mathcal{I}$.}
\end{enumerate}
\noindent FPrvR outputs all nodes in $\mathcal{I}$ and their stable routes, excluding fixed-route attackers, $d$ and all nodes output by FCR and FPeeR.

Before proceeding any further, in Lemma \ref{lem_get_stable_customer} we show that in step 2 of FPrvR we are guaranteed to find a node whose most preferred route in its $BPR$ set contains a provider route only one hop away from a node in $\mathcal{I}$ as long as there is at least one node $i \notin \mathcal{I}$ such that $BPR^i$ contains a provider route.

\begin{lem} \label{lem_get_stable_customer} Consider an arbitrary node $i \notin S_A$ such that $BPR^i$ contains a provider route $R$.  There exists at least one node $j \notin \mathcal{I}$ whose most preferred route in $BPR^j$ is a provider route and $Nxt^j \in \mathcal{I}$.\end{lem}

\begin{proof}
We consider an arbitrary routing configuration and suppose that all fixed-route attackers announce their fixed routes from the start.
If $i$ is fixed to $R$, then we are done.  Otherwise, there must be at least one node $j$ in $R$ who prefers some other route $R'$ to its prefix in $R$.  Consider the closest such node $j$ to $i$, suppose that $R'$ is $j$'s most preferred route, and note that $R'$ must be a provider route (otherwise it would have been fixed by now with FCR or FPeeR).  We can use the argument in Lemma \ref{lem_get_stable_provider}  but consider only provider routes instead of only customer routes until either a customer-provider cycle is obtained (which would contradict our routing system set up) or we reach a node whose $BPR$ set contains a direct provider route to a node in $\mathcal{I}$ that it is fixed to.  
\end{proof}

\begin{lem} \label{lem_convergence} Every node in a commercial routing system is guaranteed to stabilize to the same route it gets fixed to in FR, for any activation schedule, for any initial routing configuration. \end{lem}

\begin{proof}
Consider an arbitrary node $i$ fixed by FCR subroutine to a customer route.  Due to step $4$ of FCR, and because FCR fixes nodes to only exportable routes, $i$ must have been fixed to the best route that could ever be available to $i$ during BGP's execution for any starting configuration and any activation schedule.

Now consider an arbitrary node $i$ fixed by FPeeR subroutine to a peer route.  After execution of FCR, nodes considered in FPeeR must have no customer routes in their $BPR$ sets (step 4 of FCR), otherwise they would have been stabilized in FCR, so their favorite available routes must all be peer routes.  Therefore, due to step $4$ of FCR and step 3 of FPeeR, and because FPeeR fixes nodes to only exportable routes, $i$ must have been fixed to the best route that could ever be available to $i$ during BGP's execution for any starting configuration and any activation schedule.

Finally consider an arbitrary node $i$ fixed by FPrvR subroutine to a provider route.  After execution of FCR and then FPeeR, nodes considered in FPrvR must have neither customer nor peer routes in their $BPR$ sets (steps 4 and 3 of FCR and FPeeR respectively), otherwise they would have been stabilized in FCR or FPrvR, so their favorite available routes must all be provider routes.  Thus, due to steps $4$ of FCR and FPrvR as well as step 3 of FPeeR, and because FPrvR fixes nodes to only exportable routes, $i$ must have been fixed to the best route that could ever be available to $i$ during BGP's execution for any starting configuration and any activation schedule.
\end{proof}

In what follows, we show that during BGP's execution, all nodes who stabilize to customer, peer and provider routes, do so within $x$, $x+1$ and $2x +1$ asynchronous rounds respectively.

\begin{lem} \label{lem_stabilize_providers} All nodes in the output of FCR, stabilize within $x$ asynchronous rounds during BGP's execution. \end{lem}

\begin{proof}
 By Lemma \ref{lem_convergence}, every node who gets fixed in FCR to a customer route, converges to the same route during BGP's execution.  All customer routes that do not contain fixed route attackers, cannot be greater than $x$ in length.  Customer routes that contain fixed-route attackers can be longer than $x$, but fixed-route attackers announced their routes in the beginning of the protocol's execution.  For the sake of contradiction, suppose that there exists at least one node $i$ who requires more than $x$ asynchronous rounds to stabilize to a route $R$.  This means that there is at least one node $j$ who has selected its suffix $R_{|j}$ or who selected some other route that it does not export to $i$, only a single round before $i$ stabilized.  Otherwise $i$ could have selected $R$ earlier.  In either case, there must exist at least one node $k$ who has selected its suffix $R_{|j_{|k}}$ or who selected some other route that it does not export to $j$, only a single round before $j$ stabilized.  Otherwise $j$ could have selected $R_{|j}$ earlier.  Since $|V|$ is finite, we can continue this argument until we obtain either (i) a node whose next hop is in $S_A \cup \{d\}$, or (ii) customer-provider cycle containing no attackers.  The latter case would contradict the condition that no node can be its own indirect customer.

 Note that all routes considered in this argument are customer routes only and that the choices that nodes, that are guaranteed to converge to customer routes, make with respect to the routes that they are guaranteed to converge to are based only on the choices of their customers, whose choices depend on their customers, and so on.  This is so due to nodes' preferences and export policies in a commercial routing system.  As we follow the argument in this proof to reach case (i) we notice that for every node who converges to route whose length is more than the number of asynchronous rounds it took for that node to converge, we can extend a chain of customers that we have started at node $i$ who took more than $x$ asynchronous rounds to converge.  Thus, in case (i) we could construct a chain of customers that either does not contain an attacker and is longer than $x$ or does contain an attacker and is longer than $x+y$, where $y$ is the length of the route that the fixed attacker on that route announces.  This, however, contradicts the height of the hierarchy of the routing system, noting that attackers announce their routes at the beginning of the protocol's execution. \end{proof}

\begin{lem} \label{lem_stabilize_peers} All nodes in the output of FPeeR, stabilize within $x+1$ asynchronous rounds during BGP's execution. \end{lem}

\begin{proof}
 By Lemma \ref{lem_convergence}, every node who gets fixed in FPeeR to a peer route, converges to the same route during BGP's execution.  All peer routes that do not contain fixed route attackers, cannot be greater than $x+1$ in length.  Peer routes that contain fixed-route attackers can be longer than $x+1$, but are announced in the beginning of the protocol's execution.
 By Lemma \ref{lem_stabilize_providers}, all nodes found by FCR must be stable, and therefore are not effected by route announcements after $x$ asynchronous rounds.  Nodes considered in FPeeR must not have customer routes in their $BPR$ sets, otherwise they would have been stabilized in FCR, so their favorite available routes must all be peer routes.  Since in FPeeR the only nodes added are the ones whose $Nxt$ is in $\mathcal{I}$ and either converges to a customer route, is $d$ or a fixed attacker, thus-added nodes must stabilize within only one more round than their $Nxt$.
\end{proof}

\begin{lem} \label{lem_stabilize_customers} All nodes in the output of FPrvR, stabilize within $2x+1$ asynchronous rounds. \end{lem}

\begin{proof}
By Lemma \ref{lem_convergence}, every node who gets fixed in FPrvR to a provider route, converges to the same route during BGP's execution.
 All provider routes that do not contain fixed route attackers, cannot be greater than $2x+1$ in length.
 Provider routes that contain fixed-route attackers can be longer than $2x+1$, but are announced in the beginning of the protocol's execution.
 By Lemmas \ref{lem_stabilize_providers} - \ref{lem_stabilize_peers}, all nodes found by FCR and FPeeR must be stable, and therefore are not effected by route announcements after $x+1$ asynchronous rounds.  Nodes considered in FPrvR must not have customer or peer routes in their $BPR$ sets, otherwise they would have been stabilized in FCR or FPeeR, so their favorite available routes must all be provider routes.  For the sake of contradiction, suppose that there exists at least one node $i$ who requires more than $2x+1$ asynchronous rounds to stabilize to a route $R$.  Note that this means that it must take $i$ more than $x$ asynchronous rounds to stabilize to a route that contains a node who is $d$, an attacker, or who stabilized to a peer or a customer suffix of this route in $x+1$ or $x$ asynchronous rounds by Lemmas\ref{lem_stabilize_peers} and Lemmas\ref{lem_stabilize_providers} respectively, after that node has converged.  We will argue about only the nodes in this node's prefix of the route; these are the only nodes who stabilize to provider routes in $i$'s route.

There must at least one node $j$ who has selected its suffix $R_{|j}$ or who selected some other route that it does not export to $i$, only a single round before $i$ stabilized.  Otherwise $i$ could have selected $R$ earlier.  In either case, there must exist at least one node $k$ who has selected its suffix $R_{|j_{|k}}$ or who selected some other route that it does not export to $j$, only a single round before $j$ stabilized.  Otherwise $j$ could have selected $R_{|j}$ earlier.  Since $|V|$ is finite, we can continue this argument until we obtain either (i) a node whose next hop is in $S_A \cup \{d\}$, or (ii) customer-provider cycle containing no attackers.  The latter case would contradict the condition that no node can be its own indirect customer.

  Note that all routes considered in this argument are provider routes only and that the choices that nodes, that are guaranteed to converge to provider routes, make with respect to the routes that they are guaranteed to converge to are based only on the choices of their providers, whose choices depend on their providers, and so on.  This is so due to nodes' preferences and export policies in a commercial routing system.  As we follow the argument in this proof to reach case (i) we notice that for every node who converges to route whose length is more than the number of asynchronous rounds it took for that node to converge, we can extend a chain of providers that we have started at node $i$ who took more than $2x+1$ asynchronous rounds to converge.  Thus, in case (i) we could construct a chain of providers that either does not contain an attacker and is longer than $x$ or does contain an attacker and is longer than $x+y$, where $y$ is the length of the route that the fixed attacker on that route announces.  This, however, contradicts the height of the hierarchy of the routing system, noting that attackers announce their routes at the beginning of the protocol's execution. \end{proof}

Lemmas \ref{lem_stabilize_providers}-\ref{lem_stabilize_customers} conclude the proof of Theorem \ref{thm_GR_convergence_time} as they show that every node in a given routing system with a customer-provider hierarchy of height $x$ is guaranteed to converge to the same route within $2x+1$ asynchronous rounds regardless of the initial routing configuration and activation schedule.

\begin{scriptsize}
\bibliographystyle{abbrv} 
\bibliography{partialSec}  
\end{scriptsize}

\end{document}